\documentclass{article}
\usepackage{amsmath}
\usepackage{amssymb}
\usepackage{algorithmic}
\usepackage{algorithm}

\usepackage{amsthm}
\usepackage{array}
\usepackage{calc}
\usepackage{enumerate}
\newlength\celldim
\newlength\fontheight
\newlength\extraheight
\newcounter{sqcolumns}
\usepackage{graphicx}
\DeclareGraphicsExtensions{.pdf}
\newcolumntype{S}{
 @{}
 >{\centering \rule[-0.5\extraheight]{0pt}{\fontheight + \extraheight}%
 \begin{minipage}{\celldim}\centering}
 p{\celldim}
 <{\end{minipage}} 
 @{} }

\newcolumntype{Z}{ @{} >{\centering} p{\celldim} @{} }

\newenvironment{squarecells}[1]
  {\setlength\celldim{2em}%
   \settoheight\fontheight{A}%
   \setlength\extraheight{\celldim - \fontheight}%
   \setcounter{sqcolumns}{#1 - 1}%
   \begin{tabular}{|S|*{\value{sqcolumns}}{Z|}}\hline}
  {\end{tabular}}

\newcommand\nl{\tabularnewline\hline}
\newtheorem{thm}{Theorem}[section]

\begin{document}
\title{New Lower Bounds for the Shannon Capacity of Odd Cycles}
\author{K. Ashik Mathew and Patric R. J. \"{O}sterg\r{a}rd
% \thanks{This work was supported in part by the Academy of Finland, 
% Grant Number 132122.}%
\thanks{K. A. Mathew and P. R. J. \"{O}sterg\r{a}rd are
with the Department of Communications and Networking, 
Aalto University School of Electrical Engineering, P.O.\ Box 13000, 
00076 Aalto, Finland (e-mail: ashik.kizhakkepallathu@aalto.fi; 
patric.ostergard@aalto.fi).}}

\maketitle

\begin{abstract}
The Shannon capacity of a graph $G$ is defined as 
$c(G)=\sup_{d\geq 1}(\alpha(G^d))^{\frac{1}{d}},$ 
where $\alpha(G)$ is the independence number of $G$. 
The Shannon capacity of the cycle $C_5$ on $5$ vertices was
determined by Lov\'{a}sz in 1979, but the Shannon capacity of a cycle $C_p$
for general odd $p$ remains one of the most notorious open problems 
in information theory. By prescribing stabilizers for the 
independent sets in $C_p^d$ and using stochastic search methods, we show that 
$\alpha(C_7^5)\geq 350$, $\alpha(C_{11}^4)\geq 748$, $\alpha(C_{13}^4)\geq 1534$ and 
$\alpha(C_{15}^3)\geq 381$. This leads to improved lower bounds on the Shannon capacity
of $C_7$ and $C_{15}$: 
$c(C_7)\geq 350^{\frac{1}{5}}> 3.2271$ and $c(C_{15})\geq 381^{\frac{1}{3}}> 7.2495$.

\end{abstract}

%% \begin{IEEEkeywords}
%% Cube packing, independent set, Shannon capacity, zero-error capacity
%% \end{IEEEkeywords}
\section{Introduction}
\label{sec:intro}

The Shannon capacity of a graph is an important information-theoretic parameter
and plays a central role in the
study of the zero-error capacity of a noisy communication channel represented by the graph~\cite{KO98}. 
A communication channel transmitting $p$ different symbols can be represented by a graph
$G$ with vertex set $V$ and edge set $E$ in the following way: $V$ is the set of transmitted 
symbols, and for $v_1, v_2\in V$, $(v_1,v_2)\in E$ if the symbols $v_1$ and
$v_2$ are indistinguishable.
The \emph{Shannon capacity} of $G$ is defined as 
$$c(G)=\sup_{d\geq 1}(\alpha(G^d))^{\frac{1}{d}},$$
where $\alpha(G)$ is the independence number of $G$ and the graph strong product
is assumed \cite{S56}. For a survey of some of the early results 
related to the Shannon capacity of graphs, see~\cite{Knuthsand}.

Algebraic tools for the study of Shannon capacity were
proposed by Haemers~\cite{Haem78, HaemIEEE}
while the Shannon capacity of digraphs were investigated by
Alon~\cite{Alon98}.
See also ~\cite{Alon95, Ash87, Feig95,Szeg94} for some related studies. 

For a channel transmitting $p$ symbols represented by the elements of 
${\mathbb Z}_p=\{0,1,2,\ldots , p-1\}$ and
where two distinct symbols $s$ and $t$ are indistinguishable 
if $s - t \equiv \pm 1 \pmod p$,
the graph that represents the channel is $C_p$, the cycle on $p$ vertices.
If $p$ is even, then $c(C_p) = p/2$.
It was shown by Lov\'{a}sz~\cite{Lova1979} in 1979 
that $c(C_5)=\sqrt 5$,
but finding the Shannon capacity of $C_p$ for $p \geq 7$ and odd is still 
open~\cite{Bohmanlimit1}. 

It is well known that the independence number of $C_p^d$, the $d$th 
power (under strong product) of a cycle $C_p$,
is same as the number of hypercubes of side $2$ that can be packed in a discrete $d$-dimensional
torus of width $p$, denoted by $G(d,p)$~\cite{brsch-hypergraphs}.
See Figure~\ref{fig:2dcliq} for a visualization
of such a 2-dimensional packing and 
a corresponding independent set in $C_5^2$.
Representing independent sets as a packing of cubes often gives a 
more comprehensible model (visually) to work with.
This is especially true when we talk about symmetries, though
for the cause of adhering to formalism, we stick to a rather algebraic
notion to discuss symmetries in the remaining sections.

\begin{figure}[htbp]
\begin{center}
\begin{squarecells}{5}
  & 4 & 4 & 5 & 5 \nl 
3 & 4 & 4 &   & 3 \nl 
3 &   & 2 & 2 & 3 \nl 
1 & 1 & 2 & 2 &   \nl 
1 & 1 &   & 5 & 5 \nl 
\end{squarecells}
\vspace{5mm}
\end{center}
\begin{center}
\includegraphics{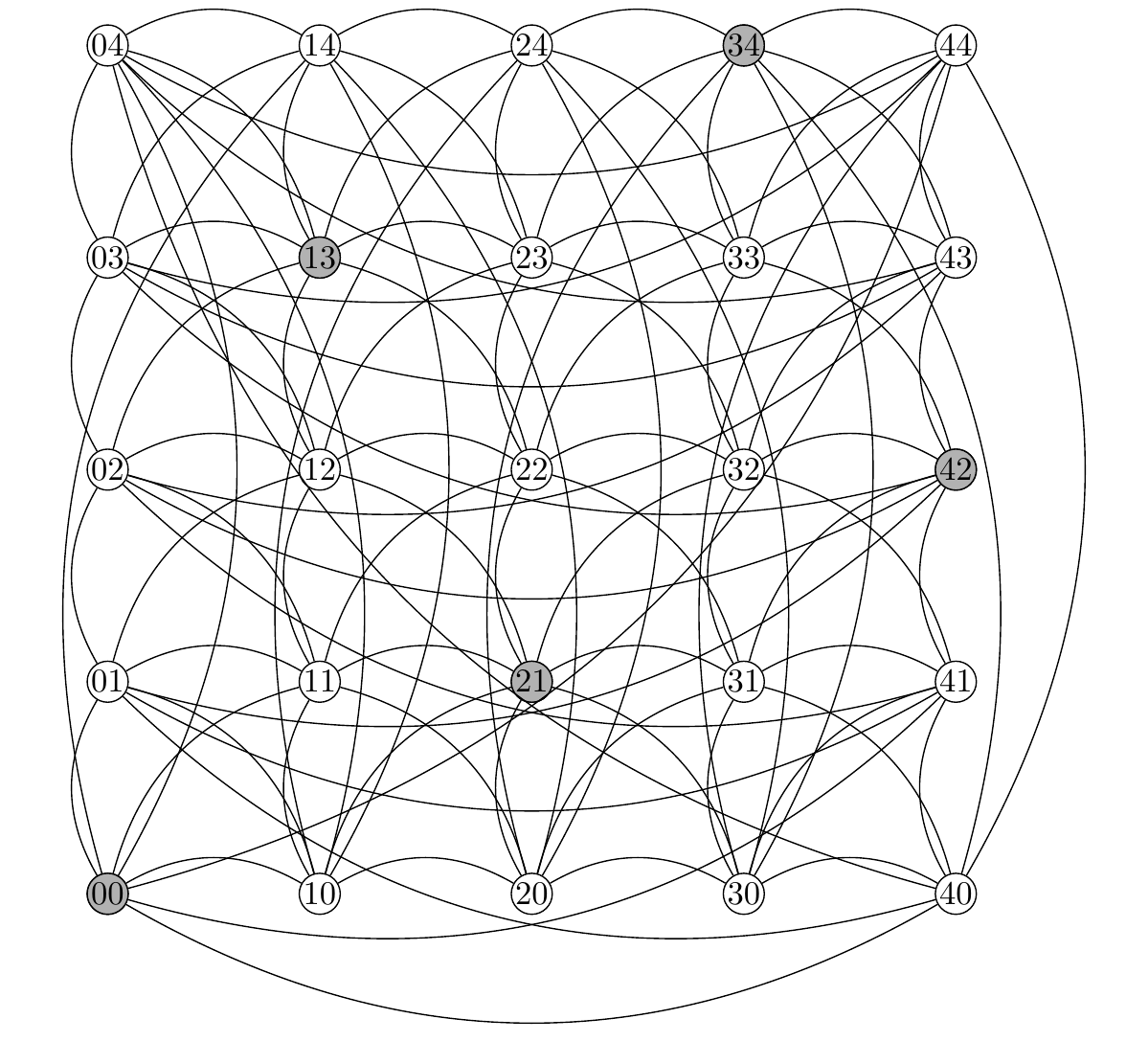}
\end{center}
\caption{Packing a torus with 2-dimensional cubes and the corresponding independent
set in $C_5^2$}
\label{fig:2dcliq}
\end{figure}

Cube packings and their different variants also form the basis of several
classical and well-studied problems in
combinatorics~\cite{brsch-hypergraphs,keller1930}. The function 
$G(d,p)$ has been studied thoroughly, and exact values and bounds
have been published in \cite{Baumertpacking} and later studies.
Several of these results have been obtained using 
exhaustive and stochastic computational methods. For example,
Baumert \textit{et al.} \cite{Baumertpacking} used exhaustive search to 
show that $G(3,7)=33$, and Vesel and \v{Z}erovnik \cite{VZ02} proved that
$G(4,7)\geq 108$ with simulated annealing. The current authors used another 
stochastic (local search) method, tabu search, to obtain lower bounds for
the capacity of triangular graphs \cite{triangular}. (Triangular graphs are
closely related to cycle graphs; the capacity problem for triangular
graphs can be studied via a generalization of the cube packing problem.)

Many of the best known cube packings possess some kind of symmetry. 
For example, the packing in Figure~\ref{fig:2dcliq} has
a symmetry generated by
$(a,b) \rightarrow (a+2,b+1)$ (addition modulo 5). 
This symmetry generates a group of order 5. Several additional 
examples can be found in the constructions of \cite{Baumertpacking}. 

In the current work, stochastic computational methods will be combined with the
idea of prescribing symmetries of packings. By prescribing symmetries, 
one is able to speed up the computer search. Obviously, such a search
has a possibility of success only if there are packings with the given
symmetries. By exploiting possible symmetries in as exhaustive manner
as possible, we are able to show that $\alpha(C_7^5)\geq 350$, 
$\alpha(C_{11}^4)\geq 748$, $\alpha(C_{13}^4)\geq 1534$ and 
$\alpha(C_{15}^3)\geq 381$. These bounds further imply that
$c(C_7)\geq 350^{\frac{1}{5}}> 3.2271$ and 
$c(C_{15})\geq 381^{\frac{1}{3}}> 7.2495$.

The paper is organized as follows. 
In Section~\ref{sec:background}, the approach of prescribing symmetries is considered,
and a stochastic local search method for finding packings is discussed in 
Section~\ref{sec:symtech}. 
In Section~\ref{sec:results}, the results are summarizes and 
tabulated. Specific packings are listed in the Appendix.

\section{Prescribing Symmetries of Independent Sets}
\label{sec:background}

The graph $G=C_p^d$ is conveniently discussed in the framework of codes.
Let $V(G)=\{0,1,\ldots, p-1\}^d$, the set of all codewords of length $d$ over 
${\mathbb Z}_p$. For $v \in V$, we denote 
$v = (v_1,v_2,\ldots ,v_d)$.  Now we define the set of edges as
\begin{equation}
\label{eq:dist}
E(G) = \{\{v,v'\} : v, v'\in V(G) \mbox{\ and\ }
\max_{1\leq i\leq d}\min\{|v_i-v'_i|,p-|v_i-v'_i|\} < 2\}.
\end{equation}
This definition further shows the one-to-one correspondence between
an independent set in $C_p^d$ and a packing in the discrete 
$d$-dimensional torus of width $p$ by (hyper)cubes of side $2$
(with their centers---or any other specific position of the cubes---in
the position given by the element of the independent set).

For small parameters, one may find the independence number of $C_p^d$ 
using Cliquer~\cite{Cliquerweb} or some other available software.
However, growing parameters makes the use of such exact algorithms
infeasible at some point. One way of handling large instances
of combinatorial search problems is to prescribe symmetries 
\cite[Chapter 9]{KO}.

Symmetries of an independent set in a graph $G$ are elements of the automorphism 
group of $G$---denoted by $\mbox{Aut}(G)$---that stabilize the
independent set. 

Let $N[v]$ denote the closed neighborhood of the vertex $v$.
A graph $G$ is called \emph{thin} if $N[u] \neq N[v]$ whenever
$u \neq v$. A \emph{prime} graph $G$ is one that cannot be written
as $G=G_1\boxtimes G_2$ (strong product of $G_1$ and $G_2$) for non-trivial graphs $G_1$ and $G_2$.

\begin{thm}[\cite{ImrBook}, Theorem 7.18]
\label{thm:imr}
For a graph $G=G_1\boxtimes G_2\boxtimes\cdots\boxtimes G_n$
where $G_1, G_2,\ldots , G_n$ are connected, thin and prime graphs,
the automorphism group of $G$ is isomorphic to the automorphism group of
 the disjoint union of graphs $G_1, G_2, \ldots, G_n$.
\end{thm}
\begin{thm}[Frucht~\cite{Frucht49}]
\label{thm:fru}
If $G$ is a connected graph and $nG$ denotes the graph representing
$n$ disjoint copies of $G$, then
$\mbox{Aut}(nG)$ is the wreath product $\mbox{Aut}(G)\wr S_n$.
\end{thm}
\begin{thm}
The automorphism group of $C_p^d$ (for $p>3$) is isomorphic to the wreath product 
$D_p\wr S_d$, where $D_p$ is the dihedral group of order $2p$ 
and $S_d$ is the symmetric group of degree $d$.
\end{thm}

\begin{proof}
We first show that $C_p$ on vertices $\{0,1,2,\ldots ,p-1\}$
and edge set $\{\{u,v\}: u-v=\pm 1\mod p\}$ 
is thin and prime for $p>3$. 
Consider any two distinct vertices $x,y\in V(C_p)$.
If $\{x,y\}\notin E(C_p)$, $x\notin N[y]$ and so $N[x]\neq N[y]$.
Suppose $\{x,y\}\in E(C_p)$. This means that (\emph{w.l.o.g.})
$x-y=1\mod p$.
Consider the vertex $z=x+1\mod p$.
Clearly, $\{x,z\}\in E(C_p)$ and $\{y,z\}\notin E(C_p)$,
 leading to $N[x]\neq N[y]$.
So, $C_p$ is thin. Now we observe that $C_p$ is prime by the
following argument.
By definition, if $G=G_1\boxtimes G_2$ is connected, both $G_1$ and
$G_2$ are connected.
Since $K_2\boxtimes K_2=K_4$, the strong product of
any two graphs with at least one edge each has $K_4$ as a subgraph.
Since $C_p$ does not have $K_4$ as a subgraph, $C_p$ is prime.
Now that $C_p$ is  connected, thin and  prime,
Theorem~\ref{thm:imr} applies to it.
Since the automorphism group of $C_p$ is the dihedral group $D_p$, 
using Theorem~\ref{thm:fru}, the result follows.
\end{proof}

The order of the group $\mbox{Aut}(C_p^d)$ is 
$|\mbox{Aut}(C_p^d)| =|D_p\wr S_d|=(2p)^dd!$.
In the framework of codes, introduced earlier, elements of $\mbox{Aut}(C_p^d)$ 
act by a permutation of the coordinates followed by permutations
of the coordinate values (separately for each coordinate) that have
the form
\begin{equation}
\label{eq:perm}
i \rightarrow ai + b\!\!\! \pmod{p},\ a \in \{-1,1\},\ b \in {\mathbb Z}_p,
\end{equation}

Given two codes corresponding to independent sets or packings of cubes,
we say that these are \emph{equivalent} if one of the codes can be obtained
from the other with a mapping in the action of $\mbox{Aut}(C_p^d)$. Such
mappings from a code onto itself---which are formally \emph{stabilizers} 
of the code (and the corresponding independent set and the packing) under
the action of $\mbox{Aut}(C_p^d)$---are
said to form the automorphism group of the code. 

The automorphism group of a code is a subgroup of $\mbox{Aut}(C_p^d)$. When prescribing
possible automorphism groups, we therefore consider subgroups of $\mbox{Aut}(C_p^d)$
up to conjugacy. Moreover, we reduce the number of
groups to consider by explicitly restricting the computations to cyclic groups. 
(In this manner we are still able to cover a large part of the groups, since 
most large groups that are omitted will have a cyclic subgroup amongst the 
groups considered.) 

Having prescribed an automorphism group of a code, the action of the group
partitions all possible codewords into orbits. In the framework of 
independent sets, we now get instances of the maximum weight independent
set problem. The vertex set consist of all \emph{admissible} orbits:
the pairs of codewords in the set must fulfill the distance criterion in
(\ref{eq:dist}). The weight of a vertex is the number of codewords in the 
orbit. Finally, edges are inserted whenever no pairs of codewords, one from each of
the orbits, violate the distance criterion in (\ref{eq:dist}).

With prescribed automorphism groups, we can extend the range of parameters
for which the running time of Cliquer (which can also handle weighted
graphs) or similar software is feasibly short. Moreover, by also changing
the computational approach from being exact to becoming stochastic, 
we can extend the range of parameters even further. Such an approach
will be discussed next.

\section{Stochastic Search for Weighted Independent Sets}
\label{sec:symtech}

The graphs obtained in the previous section are weighted. In general,
let $G$ be an arbitrary graph with vertex set $V(G)$ and edge set
$E(G)$, where each vertex has a positive integer weight. This is
obviously a generalization of the case of maximum independent sets,
which we get by letting all weights be 1. Note that since an independent
graph corresponds to a clique in the complement graph, any discussion
of independent sets apply to cliques and vice versa.
The maximum (weight) independent set problem is surveyed in \cite{BBPP99},
in the framework of cliques.

The decision problem of finding an independent set of weight at least
$k$ in a graph $G$ is NP-complete, so no polynomial-time general algorithms
are expected to be discovered. Due to the hardness (and the importance) 
of the problem, a lot of effort has been put on developing stochastic 
algorithms. For unweighted graphs and with stochastic algorithms, the main 
approach has been to process an independent set by adding and removing
vertices. Unfortunately, such a straightforward approach is not as effective
when the vertices have different weights.

Montemanni and Smith~\cite{MS09} discovered a technique for modifying
independent sets (in terms of cliques) by removing not one but many
vertices at a time. After removing a set of vertices, an exact algorithm
(like Cliquer) can be used to find a set of vertices to add that have
the largest possible weight. In some sense, this approach lies in between
basic stochastic search and exact algorithms.

The main decision to be made in the approach by Montemanni and Smith is
the set of vertices to remove from an independent set. In \cite{MS09}
vertices are removed in a random manner. When one thinks about
this problem in the context of cube packings, removal means removing
(hyper)cubes. When cubes are removed, there will be holes in the packing.
But with such holes that are not connected, we will have a situation
equivalent to that of sequentially removing a smaller number of cubes
in different parts of the packing. 

The second author \cite{O15} realized that with instances of the
maximum weight independent set problem that come from packing problems,
one may remove one vertex $v$ and all vertices that are within a certain
heuristic distance from $v$. The heuristic distance, which does not
have to be a metric, is defined separately
for each pair of vertices of a graph, and some experimenting is
typically needed to find a proper definition. The approach in \cite{O15}
has been used to find new $q$-analog packings \cite{BOW15}.

We here use the algorithm developed in \cite{O15} and define the
distance between vertices of the weighted graphs as the minimum
of 
\[
d(v,v') = \sum_{i=1}^d\min\{|v_i-v'_i|,p-|v_i-v'_i|\}.
\]
over all pairs of codewords, with one codeword from each orbit.

\section{Results}
\label{sec:results}

By applying the approaches discussed in this paper and using more than
2 CPU-years in the stochastic search, 
we have obtained independent sets that attain the following bounds for $G(d,p)$:
$G(5,7)\geq 350$, $G(4,11)\geq 748$, $G(4,13)\geq 1534$ and 
$G(3,15)\geq 381$. These also leads to improved lower bounds on the Shannon 
capacity of $C_7$ and $C_{15}$: 
$c(C_7)\geq 350^{\frac{1}{5}}> 3.2271$ and $c(C_{15})\geq 381^{\frac{1}{3}}> 7.2495$.
There previous best known lower bounds for $c(C_7)$ and $c(C_{15})$
were
$108^{\frac{1}{4}}> 3.2237$~\cite{VZ02} and $380^{\frac{1}{3}}>
7.2431$~\cite{Baumertpacking} respectively.
The best known lower bounds for $c(C_p)$ for other small cycles
of odd length are: $c(C_9)\geq 81^{\frac{1}{3}}>
4.3267$~\cite{Baumertpacking},
$c(C_{11})\geq 148^{\frac{1}{3}}> 5.2895$~\cite{Baumertpacking}
and $c(C_{13})\geq 247^{\frac{1}{3}}> 6.2743$~\cite{Bohman1}.
The Lov\'{a}sz's $\vartheta$-function
\begin{equation*}
\vartheta(p)=\frac{p\cos\frac{\pi}{p}}{1+\cos\frac{\pi}{p}}
\end{equation*}
gives
upper bounds for the Shannon capacity of the odd cycle $C_p$~\cite{Lova1979}.
Using this function,  
$c(C_7)< 3.3177$, $c(C_9)< 4.3601$,
$c(C_{11})< 5.3864$, $c(C_{13})< 6.4042$
and $c(C_{15})< 7.4172$.

The currently best known upper and lower bounds for $G(d,p)$ are listed in
Table~\ref{tab:bounds} together with keys.
Only one key is provided
in cases where the value can be obtained using more than one method.

\begin{table}[htbp]
\begin{center}
\caption{Bounds on $G(d,p)$}

\label{tab:bounds}
\begin{tabular}{cccccc}\\\hline
$p\backslash d$&1&2&3&4&5\\
\hline
$5$ & $^a2^a$ & $^a5^a$ & $^c10^f$ & $^c25^d$ & $^c50$--$55^j$\\
$7$ & $^a3^a$ &$^a10^a$ & $^f33^f$ & $^h108$--$115^d$ & $^k350$--$401^j$\\
$9$ & $^a4^a$ &$^a18^a$ & $^e81^d$ & $^c324$--$361^j$ & $^c1458$--$1575^j$\\
$11$ & $^a5^a$ &$^a27^a$ & $^e148^d$ & $^k748$--$814^d$ & $^c3996$--$4477^d$\\
$13$ & $^a6^a$ &$^a39^a$ & $^g247^i$ & $^k1534$--$1605^d$ & $^c9633$--$10432^d$\\
$15$ & $^a7^a$ &$^a52^a$ & $^k381$--$390^d$ & $^b2720$--$2925^d$ & $^c19812$--$21937^d$\\\hline
\end{tabular}
\end{center}
\end{table}

\vspace{5mm}
\noindent
Key to Table~\ref{tab:bounds}.\\
\begin{tabular}{ll}
%\multicolumn{2}{@{}l}{Unmarked bounds are from Theorems \ref{thm:one} and \ref{thm:two}.}\\
\multicolumn{2}{@{}l}{Bounds:}\\

$^a$ & $G(1,p)=\lfloor \frac{p}{2}\rfloor , G(2,p) =
      \lfloor\frac{p^2-p}{4}\rfloor $~\cite[Theorem 2]{Baumertpacking}\\
$^b$ & $G(d,p)\geq 1 +
       G(d,p-2) \frac{p^d-2^d}{(p-2)^d}$~\cite[Corollary 2]{Baumertpacking}\\
$^c$ & $G(d,p)\geq G(d_1,p) G(d-d_1,p)$~\cite[Corollary 3]{Baumertpacking}\\
$^d$ & $G(d,p) \leq \lfloor \frac{p}{2} G(d-1,p)\rfloor$~\cite[Lemma 2]{Baumertpacking}\\
$^e$ & Baumert \textit{et al.}~\cite[Theorem 3]{Baumertpacking}\\
$^f$ & Baumert \textit{et al.}~\cite[Theorem 4]{Baumertpacking}\\
$^g$ & Baumert \textit{et al.}~\cite[Theorem 6]{Baumertpacking}\\
$^h$ & Vesel and \v{Z}erovnik \cite{VZ02}\\
$^i$ & Bohman, Holzman, and Natarajan~\cite{Bohman1}\\
$^j$ & $G(d,p) \leq
       \left [ \frac{p\cos\frac{\pi}{p}}{1+\cos\frac{\pi}{p}}\right ]^d$~\cite{Lova1979}\\
$^k$ & This paper, see Appendix\\

\end{tabular} 

\bibliographystyle{IEEETran}

\section*{Appendix}
\begin{small}

We here list codes giving the four new lower bounds. 
The permutation of coordinates is the identity
permutation in all generators of the groups, and $a=1$ for all value
permutations in
(\ref{eq:perm}). We
therefore present the groups by simply listing
the values of $b$ for the $d$ value permutations of a generator.

\subsection*{$G(5,7)\geq 350$:}
\noindent
Generator: (0, 1, 1, 5, 1)\\
Group order: 7\\
Orbit representatives: (0, 5, 6, 6, 0), (0, 0, 6, 6, 0), (3, 3, 0, 6, 0), (0, 5, 2, 1, 0), (2, 5, 6, 5, 0), (1, 3, 0, 0, 0), (2, 3, 2, 0, 0), (2, 1, 0, 4, 0), (2, 5, 2, 1, 0), (0, 2, 1, 2, 0), (4, 2, 6, 3, 0), (4, 0, 0, 3, 0), (5, 1, 2, 3, 0), (3, 6, 1, 5, 0), (4, 5, 0, 2, 0), (3, 4, 5, 2, 0), (1, 6, 1, 6, 0), (2, 3, 6, 4, 0), (5, 5, 1, 4, 0), (5, 3, 1, 3, 0), (6, 4, 0, 1, 0), (0, 0, 2, 2, 0), (6, 0, 1, 0, 0), (5, 1, 5, 0, 0), (5, 6, 6, 0, 0), (5, 3, 5, 1, 0), (0, 3, 6, 5, 0), (2, 0, 2, 1, 0), (4, 0, 3, 5, 0), (4, 4, 2, 6, 0), (4, 2, 3, 6, 0), (1, 5, 5, 3, 0), (6, 2, 4, 5, 0), (4, 4, 4, 6, 0), (6, 2, 0, 0, 0), (1, 0, 5, 4, 0), (4, 6, 4, 0, 0), (3, 1, 4, 1, 0), (3, 6, 5, 2, 0), (6, 0, 5, 4, 0), (2, 3, 3, 2, 0), (1, 1, 4, 2, 0), (1, 5, 3, 3, 0), (1, 2, 4, 4, 0), (1, 1, 1, 6, 0), (3, 1, 1, 6, 0), (0, 3, 3, 2, 0), (6, 4, 3, 4, 0), (6, 6, 3, 4, 0), (6, 5, 5, 4, 0)
\subsection*{$G(4,11)\geq 748$:}
\noindent
Generator: (1, 5, 8, 9)\\
Group order: 11\\
Orbit representatives: (9, 10, 0, 0), (7, 10, 9, 0), (5, 4, 0, 0), (5, 4, 2, 0), (2, 3, 0, 0), (7, 4, 2, 0), (7, 8, 1, 0), (1, 10, 2, 0), (2, 7, 4, 0), (0, 7, 5, 0), (6, 9, 7, 0), (6, 6, 1, 0), (8, 6, 1, 0), (8, 8, 10, 0), (4, 2, 1, 0), (5, 2, 3, 0), (6, 4, 4, 0), (5, 2, 5, 0), (3, 7, 6, 0), (2, 9, 6, 0), (4, 9, 6, 0), (1, 9, 4, 0), (1, 7, 7, 0), (5, 7, 8, 0), (6, 6, 10, 0), (3, 2, 3, 0), (8, 4, 4, 0), (3, 7, 8, 0), (10, 10, 2, 0), (0, 5, 5, 0), (9, 6, 5, 0), (4, 9, 8, 0), (4, 7, 10, 0), (0, 10, 0, 0), (7, 2, 4, 0), (7, 2, 6, 0), (7, 0, 7, 0), (6, 8, 10, 0), (0, 3, 10, 0), (8, 6, 3, 0), (10, 6, 3, 0), (3, 5, 0, 0), (1, 1, 1, 0), (0, 5, 7, 0), (2, 5, 7, 0), (2, 3, 9, 0), (1, 5, 9, 0), (3, 5, 9, 0), (3, 0, 1, 0), (4, 0, 3, 0), (2, 0, 4, 0), (3, 9, 4, 0), (4, 0, 5, 0), (6, 0, 5, 0), (9, 8, 1, 0), (10, 1, 8, 0), (8, 1, 9, 0), (1, 1, 10, 0), (10, 1, 10, 0), (0, 8, 2, 0), (9, 8, 3, 0), (9, 1, 6, 0), (10, 3, 6, 0), (8, 4, 6, 0), (5, 0, 7, 0), (1, 3, 7, 0), (10, 3, 8, 0), (9, 10, 9, 0) 
\subsection*{$G(4,13)\geq 1534$:}
\noindent
Generator: (0, 1, 0, 2)\\
Group order: 13\\
Orbit representatives: (9, 6, 7, 0), (9, 8, 9, 0), (7, 1, 8, 0), (0, 7, 6, 0), (8, 10, 8, 0), (7, 11, 0, 0), (4, 11, 11, 0), (5, 4, 2, 0), (8, 12, 9, 0), (3, 7, 10, 0), (5, 7, 10, 0), (5, 0, 1, 0), (2, 0, 12, 0), (1, 7, 8, 0), (1, 7, 10, 0), (0, 7, 4, 0), (12, 0, 9, 0), (5, 2, 2, 0), (11, 2, 7, 0), (11, 0, 5, 0), (10, 4, 5, 0), (1, 5, 5, 0), (3, 8, 2, 0), (4, 0, 12, 0), (11, 0, 7, 0), (5, 12, 5, 0), (3, 5, 7, 0), (5, 12, 7, 0), (7, 12, 7, 0), (8, 12, 11, 0), (2, 12, 3, 0), (2, 3, 4, 0), (4, 10, 4, 0), (7, 2, 1, 0), (0, 9, 5, 0), (0, 9, 7, 0), (1, 5, 7, 0), (10, 4, 7, 0), (1, 1, 2, 0), (6, 10, 4, 0), (6, 10, 6, 0), (12, 11, 5, 0), (12, 11, 7, 0), (11, 2, 9, 0), (11, 4, 9, 0), (6, 3, 8, 0), (6, 5, 9, 0), (7, 1, 10, 0), (7, 3, 10, 0), (5, 9, 10, 0), (9, 10, 10, 0), (6, 5, 11, 0), (10, 1, 3, 0), (4, 5, 9, 0), (0, 11, 9, 0), (2, 11, 11, 0), (7, 4, 2, 0), (4, 6, 3, 0), (6, 6, 3, 0), (9, 12, 3, 0), (0, 0, 0, 0), (8, 1, 12, 0), (6, 7, 12, 0), (6, 9, 12, 0), (9, 10, 12, 0), (8, 8, 3, 0), (6, 8, 4, 0), (8, 10, 4, 0), (11, 2, 5, 0), (7, 6, 5, 0), (9, 6, 5, 0), (9, 8, 5, 0), (7, 8, 6, 0), (8, 10, 6, 0), (9, 8, 7, 0), (12, 1, 2, 0), (1, 3, 2, 0), (0, 12, 2, 0), (3, 9, 11, 0), (2, 3, 7, 0), (4, 3, 8, 0), (2, 9, 9, 0), (2, 11, 9, 0), (7, 0, 1, 0), (9, 1, 1, 0), (9, 12, 1, 0), (4, 8, 4, 0), (5, 1, 8, 0), (11, 0, 0, 0), (11, 11, 0, 0), (9, 10, 1, 0), (11, 12, 2, 0), (10, 3, 3, 0), (12, 10, 3, 0), (0, 2, 0, 0), (11, 2, 0, 0), (0, 9, 9, 0), (12, 2, 11, 0), (11, 4, 11, 0), (1, 9, 11, 0), (0, 11, 11, 0), (2, 2, 0, 0), (4, 2, 0, 0), (2, 4, 0, 0), (4, 4, 0, 0), (0, 11, 0, 0), (3, 6, 1, 0), (3, 10, 2, 0), (2, 1, 4, 0), (3, 12, 5, 0), (2, 1, 6, 0), (5, 1, 6, 0), (4, 3, 6, 0), (3, 7, 8, 0), (10, 6, 9, 0), (12, 0, 11, 0), (10, 6, 11, 0), (10, 8, 11, 0)
\subsection*{$G(3,15)\geq 381$:}
\noindent
Generator: (5, 0, 10)\\
Group order: 3\\
Orbit representatives: (1, 10, 4), (1, 11, 0), (10, 10, 0), (1, 2, 2), (13, 3, 2), (9, 11, 2), (2, 8, 4), (7, 11, 2), (1, 0, 2), (3, 11, 0), (5, 11, 1), (1, 10, 2), (3, 10, 4), (0, 12, 2), (12, 10, 1), (14, 10, 1), (10, 9, 2), (9, 7, 3), (8, 9, 1), (10, 9, 4), (7, 11, 4), (9, 11, 4), (7, 7, 1), (7, 7, 3), (8, 5, 3), (11, 12, 1), (13, 12, 1), (14, 0, 0), (6, 9, 1), (3, 11, 2), (6, 9, 3), (4, 8, 4), (6, 5, 0), (2, 4, 3), (1, 0, 0), (12, 1, 1), (1, 6, 1), (5, 7, 2), (2, 8, 2), (4, 9, 2), (1, 6, 3), (3, 6, 4), (4, 4, 2), (6, 5, 2), (12, 14, 1), (11, 12, 3), (13, 12, 3), (12, 5, 4), (4, 5, 0), (5, 7, 0), (2, 9, 0), (4, 9, 0), (3, 6, 2), (11, 7, 4), (11, 7, 2), (12, 6, 0), (11, 8, 0), (14, 6, 1), (0, 8, 1), (8, 9, 3), (9, 3, 3), (11, 3, 3), (4, 4, 4), (6, 4, 4), (3, 7, 0), (13, 8, 0), (13, 8, 2), (14, 6, 3), (0, 8, 3), (12, 10, 3), (14, 10, 3), (3, 2, 3), (5, 3, 0), (7, 3, 0), (10, 1, 1), (5, 2, 4), (10, 14, 3), (8, 13, 4), (1, 2, 0), (14, 2, 0), (13, 4, 0), (0, 4, 1), (2, 4, 1), (0, 4, 3), (5, 0, 1), (14, 14, 2), (12, 14, 3), (0, 12, 4), (6, 13, 2), (5, 0, 3), (5, 11, 3), (6, 13, 4), (9, 3, 1), (11, 3, 1), (8, 5, 1), (9, 7, 1), (10, 5, 2), (12, 5, 2), (10, 5, 4), (5, 6, 4), (9, 12, 0), (10, 14, 1), (8, 13, 2), (3, 0, 3), (4, 13, 3), (0, 13, 0), (2, 13, 0), (4, 13, 1), (5, 2, 2), (7, 2, 2), (9, 1, 3), (7, 0, 4), (7, 2, 4), (3, 0, 1), (3, 2, 1), (14, 1, 2), (12, 1, 3), (1, 1, 4), (14, 1, 4), (13, 3, 4), (8, 1, 0), (8, 14, 0), (7, 0, 2), (2, 13, 2), (2, 12, 4), (1, 14, 4), (14, 14, 4)
\end{small}

% ;;;;
\end{document}